\documentstyle[12pt,my_subfigure,onecolumn,a4,amssymb]{IEEEtran}

\input{epic.sty}
\input{eepic.sty}

\newcommand{\subpostscript}[2]{
    \setlength{\epsfxsize}{#2}
    \epsfbox{#1}
}

\input{psfig}               
\input{epsf}                

\newcommand{\bthm} {\begin{thm} }
\newcommand{\ethm} {\end{thm}}
\newcommand{\blem} {\begin{lem} }
\newcommand{\elem} {\end{lem}}
\newcommand{\bcor} {\begin{cor} }
\newcommand{\ecor} {\end{cor}}

\newcommand{\beq}  {\begin{equation}}
\newcommand{\eeq}  {\end{equation}}
\newcommand{\beqa}{\begin{eqnarray}}
\newcommand{\eeqa} {\end{eqnarray}}
\newcommand{\bdis} {\begin{displaymath}}
\newcommand{\edis} {\end{displaymath}}
\newcommand{\bitem}{\begin{itemize}}
\newcommand{\eitem}{\end{itemize}}
\newcommand{\ba}  {\begin{array}}
\newcommand{\ea}  {\end{array}}

\newcommand{\rmS}  {{\mathrm{S}}}
\newcommand{\rmR}  {{\mathrm{R}}}

\newcommand{\TDBC}  {{\mathrm{TDBC}}}
\newcommand{\ANC}  {{\mathrm{ANC}}}
\newcommand{\outage}  {{\mathrm{outage}}}

\newtheorem{theorem}{Theorem}
\newtheorem{cor}{Corollary}
\newtheorem{lemma}{Lemma}

\renewcommand{\QED}{\QEDopen}

\begin{document}

\begin{titlepage}

\title{Finite-SNR Diversity-Multiplexing Tradeoff and Optimum Power
Allocation in Bidirectional Cooperative Networks}

\author{\vspace{1cm}Zhihang~Yi and Il-Min~Kim, {\it Senior Member, IEEE}\\
\vspace{5mm}
Department of Electrical and Computer Engineering\\
Queen's University\\
Kingston, Ontario, K7L 3N6\\
Canada\\
\vspace{5mm} Email: ilmin.kim@queensu.ca}

\maketitle

\vspace{15mm}

\end{titlepage}
\begin{abstract}

This paper focuses on analog network coding (ANC) and time
division broadcasting (TDBC) which are two major protocols used in
bidirectional cooperative networks. Lower bounds of the outage
probabilities of those two protocols are derived first. Those
lower bounds are extremely tight in the whole signal-to-noise
ratio (SNR) range irrespective of the values of channel variances.
Based on those lower bounds, finite-SNR diversity-multiplexing
tradeoffs of the ANC and TDBC protocols are obtained. Secondly, we
investigate how to efficiently use channel state information (CSI)
in those two protocols. Specifically, an optimum power allocation
scheme is proposed for the ANC protocol. It simultaneously
minimizes the outage probability and maximizes the total mutual
information of this protocol. For the TDBC protocol, an optimum
method to combine the received signals at the relay terminal is
developed under an equal power allocation assumption. This method
minimizes the outage probability and maximizes the total mutual
information of the TDBC protocol at the same time.

\end{abstract}

\section{Introduction}\label{sec:intro}

In traditional unidirectional cooperative networks, several relays
assist in the communication between one source and one destination
in order to achieve spatial diversity \cite{laneman1}. The {\it
diversity-multiplexing tradeoff}, which is one of the most
fundamental properties of any communication systems, of such
networks has been extensively studied
\cite{azarian}--\cite{yuksel}. It has been shown that, in the
unidirectional cooperative networks, the half-duplex constraint of
every terminal induces a severe loss of bandwidth efficiency as
demonstrated by a pre-log factor $1/2$ in the mutual information
expression. In order to overcome this difficulty, bidirectional
cooperative networks were studied in \cite{rankov}, where two
sources exchanged information with the help of several relays. As
a result, there were two traffic flows in a bidirectional
cooperative network and they were supported by the same physical
channels concurrently. Although each traffic flow still had the
pre-log factor $1/2$ in its mutual information expression, the
total mutual information of the network, which was the summation
of the mutual information of both traffic flows, no longer
suffered from the pre-log factor $1/2$. Therefore, bidirectional
cooperative networks had much higher bandwidth efficiency than
unidirectional cooperative networks.

Recently, many novel protocols were studied in the context of
bidirectional cooperative networks, such as {\it physical layer
network coding} (PNC) \cite{shengli1}--\cite{rankov2}, {\it analog
network coding} (ANC) \cite{katti2}, \cite{shengli2}, and {\it
time division broadcast} (TDBC) \cite{shengli1} and
\cite{kim1}--\cite{liu}\footnote{Note that the TDBC protocol was
called the {\it straightforward network coding} scheme in
\cite{shengli1}.}. Those protocols not only achieved high
bandwidth efficiency but also successfully controlled the
interferences between the two traffic flows in a bidirectional
cooperative network. Many previous works assumed that there was no
direct channel between the two sources \cite{shengli1},
\cite{katti2}, \cite{shengli2}, and \cite{liu}--\cite{haoy}. Under
this assumption, the authors showed that the PNC and ANC protocols
had larger capacity regions and higher sum-rates than the TDBC
protocol \cite{shengli1}, \cite{katti2}, and \cite{haoy}. On the
other hand, since there is no direct channel, the maximum
diversity gains of the PNC, ANC, and TDBC protocols are the same,
just one. Therefore, one may conclude that the PNC and ANC
protocols always outperform the TDBC protocol in terms of
diversity-multiplexing tradeoff \cite{zheng}.

In fact, numerous previous publications on cooperative networks
have also considered the case that there is a direct channel
between the two sources \cite{laneman1}--\cite{laneman2},
\cite{yang}, \cite{yuksel}, \cite{kim1}, \cite{kim2}, and
\cite{anghel1}--\cite{katti}. For this case, the comparison of
diversity-multiplexing tradeoff might have different result. It is
not hard to see that the TDBC protocol can utilize this direct
channel \cite{kim1}, \cite{kim2}; but the PNC and ANC protocols
can not do so due to the half-duplex constraint \cite{katti}.
Thus, the TDBC protocol might have a higher diversity gain than
the PNC and ANC protocols. This is fundamentally different from
the case where the direct channel does not exist. Therefore, it is
necessary and interesting to compare those protocols in terms of
diversity-multiplexing tradeoff under the assumption that the
direct channel exists. Furthermore, it is more desirable to
compare the tradeoff at finite signal-to-noise ratio (SNR) range
as in \cite{nara} instead of only at infinite SNR in \cite{zheng},
because practical communication systems, such as wireless local
area networks, usually work in the SNR range $3$--$20$ dB.
However, such comparison has not been investigated in previous
publications. This has motivated our work.

In this paper, we assume that the relays work in the
amplify-and-forward mode, and hence, we do not consider the PNC
protocol. In fact, it has been shown that the PNC protocol may
suffer considerable performance loss in fading channels
\cite{haoy}, \cite{dingz}. Thus, we focus on the ANC and TDBC
protocols in this paper. The contributions of this paper are
summarized as follows:

\begin{itemize}
    \item We derive lower bounds of the
    outage probabilities of the ANC and TDBC protocols. Those bounds are extremely tight in the whole SNR range,
    irrespective of the values of channel variances.

    \item We obtain the finite-SNR diversity-multiplexing
    tradeoffs of the ANC and TDBC protocols. Those tradeoffs
    establish a framework which enables us to make a comprehensive comparison
    between those two protocols. For example, the maximum
    diversity gain of the TDBC protocol is twice larger than that of
    the ANC protocol. However, the ANC protocol achieves a higher diversity
    gain when the multiplexing gain is larger than $1/2$, irrespective
    of the value of SNR.

    \item For the ANC protocol, we propose an optimum power allocation scheme which
    simultaneously minimizes the outage probability and maximizes
    the total mutual information of this protocol.

    \item For the TDBC protocol, we develop an optimum method to combine the received signals at the relay when
    equal power allocation is applied. This method simultaneously minimizes the
    outage probability and maximizes the total mutual information as well.
\end{itemize}

The rest of this paper is organized as follows. Section
\ref{sec:sys} describes the system models of the ANC and TDBC
protocols. Section \ref{sec:prob} focuses on the outage
probabilities and finite-SNR diversity-multiplexing tradeoffs of
those two protocols. In Section \ref{sec:power}, we first develop
an optimum power allocation scheme for the ANC protocol. Then we
propose an optimum method for the TDBC protocol to combine the
received signals at the relay. Section \ref{sec:numeric} presents
some numerical results and Section \ref{sec:conl} concludes this
paper.


\section{System Model}\label{sec:sys}

We consider a bidirectional cooperative network with two sources
and one relay, where the sources intend to exchange information
with the help of the relay. Every terminal has only one antenna
and is half-duplex. We use $\rmS_1$, $\rmS_2$, and $\rmR$ to
denote the first source, the second source, and the relay,
respectively. Let $g$ represent the fading coefficient of the
channel between $\rmS_1$ and $\rmS_2$, $h$ the channel between
$\mathrm{S}_1$ and $\mathrm{R}$, and $f$ the channel between
$\rmR$ and $\mathrm{S}_2$. Furthermore, we assume that $g$, $h$,
and $f$ are complex Gaussian random variables with zero mean and
variances $\Omega_g$, $\Omega_h$, and $\Omega_f$, respectively.
The additive noise associated with every channel is assumed to be
a complex Gaussian random variable with zero mean and unit
variance. Let $x_1$ and $x_2$ denote the information-bearing
symbol transmitted from $\rmS_1$ and $\rmS_2$, respectively. Both
$x_1$ and $x_2$ have unit power. The total transmission power of
the bidirectional cooperative network is constrained to be $3E$.


Since the two sources intend to exchange information, there are
two traffic flows in this bidirectional cooperative network. One
is from $\rmS_1$ via $\rmR$ to $\rmS_2$ and the other is from
$\rmS_2$ via $\rmR$ to $\rmS_1$. Each traffic flow can be seen as
a traditional unidirectional cooperative network. For example, in
the first traffic flow, $\rmS_1$ is the transmitter, and $\rmR$
and $\rmS_2$ are the receivers. Hence, it is reasonable to assume
that $\rmR$ knows $h$ and $\rmS_2$ knows $h$, $f$, and $g$ as in
the conventional unidirectional cooperative networks
\cite{laneman1}, \cite{anghel1}. Similarly, due to the existence
of the second traffic flow, it is reasonable to assume that $\rmR$
knows $f$ and $\rmS_1$ knows $h$, $f$, and $g$. In all, we assume
that the two sources know $h$, $f$, and $g$, and the relay knows
$h$ and $f$ as in many previous publications \cite{rankov},
\cite{lee1}, \cite{lee2}, and \cite{kuhn}.

\subsection{Analog Network Coding (ANC)}

The ANC protocol has received lots of attention recently
\cite{katti2}, \cite{shengli2}, and it is illustrated in Fig.\
\ref{fig:model}(a). In this protocol, $\rmS_1$ and $\rmS_2$
simultaneously transmit to $\rmR$ at the first time slot with
power $E$. Thus, the received signal $y_R(1)$ of $\rmR$ at the
first time slot is given by
\begin{equation}
y_R(1) = \sqrt{E}hx_1+\sqrt{E}fx_2+n_R(1),
\end{equation}
where $n_R(1)$ is the additive Gaussian noise. At the second time
slot, $\rmR$ amplifies $y_R(1)$ with an amplifying coefficient
$\rho$ and then transmits it to $\rmS_1$ and $\rmS_2$ also with
power $E$.\footnote{The setting that every terminal has the same
transmission power $E$ does not make our analysis of outage
probability lose generality. This is because, in order to obtain a
general analysis of a cooperative network, it is sufficient to
make the average SNR of every channel different as shown in
\cite{ribeiro}. Although we make the transmission power at every
terminal the same, the variances $\Omega_h$, $\Omega_f$, and
$\Omega_g$ of the channels are different in general. As a result,
the average SNR of every channel is different, which makes our
analysis still general. Furthermore, the setting that every
terminal has the same transmission power $E$ does not affect the
diversity-multiplexing tradeoff analysis neither as shown in
\cite{azarian}, \cite{zheng}.} Consequently, the received signal
$y_{S_1}(2)$ of $\rmS_1$ at the second time slot is given by
\begin{eqnarray}
y_{S_1}(2)&=&\sqrt{E}\rho hy_R(1)+n_{S_1}(2)\\
&=&\rho h^2 Ex_1+ \rho hfEx_2 +\rho h\sqrt{E} n_R(1)+n_{S_1}(2),
\end{eqnarray}
where $n_{S_1}(2)$ is the additive Gaussian noise. In order to
ensure that the transmission power at $\rmR$ is always $E$, the
amplifying coefficient is chosen as
\begin{equation}
\rho = \sqrt{\frac{1}{E|h|^2+E|f|^2+1}}.
\end{equation}

Note that the received signal $y_{S_1}(2)$ contains both $x_1$ and
$x_2$, where only $x_2$ is the desired signal and $x_1$ is
actually an interference to $\rmS_1$. Since $\rmS_1$ perfectly
knows $x_1$, it can completely remove $x_1$ from $y_{S_1}(2)$ and
obtain a new signal $\tilde{y}_{S_1}(2)$ given by
\begin{equation}
\tilde{y}_{S_1}(2) = \rho hfEx_2 +\rho h\sqrt{E}
n_R(1)+n_{S_1}(2).
\end{equation}
Consequently, the instantaneous SNR $\gamma_1^{\ANC}$ at $\rmS_1$
is given by
\begin{equation}\label{eqn:snr1anc}
\gamma_1^{\ANC} = \frac{\rho^2|hf|^2E^2}{E\rho^2|h|^2+1}
\approx \frac{E|hf|^2}{2|h|^2+|f|^2},
\end{equation}
where the approximation is by letting $\rho\approx
\sqrt{1/(E|h|^2+E|f|^2)}$. Since this approximation is very tight
in the whole SNR range, it has been used in many previous
publications \cite{anghel1}, \cite{ribeiro}. Similarly, the
instantaneous SNR $\gamma_2^{\ANC}$ at $\rmS_2$ is given by
\begin{equation}\label{eqn:snr2anc}
\gamma_2^{\ANC} = \frac{\rho^2|hf|^2E^2}{E\rho^2|f|^2+1}
\approx \frac{E|hf|^2}{2|f|^2+|h|^2}.
\end{equation}

\subsection{Time Division Broadcast (TDBC)}

In order to accomplish information exchange between $\rmS_1$ and
$\rmS_2$, the TDBC protocol, as illustrated in Fig.\
\ref{fig:model}(b), was studied in \cite{kim1},
\cite{kim2}.\footnote{In \cite{shengli1} and
\cite{kim1}--\cite{liu}, the TDBC protocol was actually studied
when the relay worked in the decode-and-forward mode. However, it
is very simple to extend it to the case that the relay works in
the amplify-and-forward mode as shown in this paper.} As in the
ANC protocol, we still assume that every terminal has the same
transmission power $E$. At the first time slot, $\rmS_1$ transmits
$x_1$ to $\rmR$ and $\rmS_2$. Thus, the received signals $y_R(1)$
of $\rmR$ and $y_{S_2}(1)$ of $\rmS_2$ at the first time slot are
given by
\begin{eqnarray}
y_R(1)=\sqrt{E}hx_1+n_R(1),&& y_{S_2}(1)=\sqrt{E}gx_1+n_{S_2}(1),
\end{eqnarray}
where $n_R(1)$ and $n_{S_2}(1)$ are the additive Gaussian noises.
At the second time slot, $\rmS_2$ transmits $x_2$ to $\rmR$ and
$\rmS_1$. The received signals $y_R(2)$ of $\rmR$ and $y_{S_1}(2)$
of $\rmS_1$ at the second time slot are given by
\begin{eqnarray}
y_R(2)=\sqrt{E}fx_2+n_R(2),&& y_{S_1}(2)=\sqrt{E}gx_2+n_{S_1}(2),
\end{eqnarray}
where $n_R(2)$ and $n_{S_1}(2)$ are the additive Gaussian noises.
At the third time slot, $\rmR$ combines $y_R(1)$ and $y_R(2)$ at
first. Then it broadcasts the combined signal to $\rmS_1$ and
$\rmS_2$. The combined signal at $\rmR$ is denoted by $x_R$ and it
is given by
\begin{equation}
x_R = \eta_1y_R(1) + \eta_2y_R(2),
\end{equation}
where
\begin{eqnarray}
\eta_1 = \sqrt{\frac{\xi }{E|h|^2+1}},&&\eta_2 =
\sqrt{\frac{(1-\xi) }{E|f|^2+1}}.
\end{eqnarray}
The choice of $\eta_1$ and $\eta_2$ ensures that $x_R$ always has
unit power. The value of $\xi$ decides how the relay combines the
signals received from two different sources and it can be used to
optimize the performance of the TDBC protocol, which will be
discussed in detail later.

At the third time slot, the received signal $y_{S_1}(3)$ of
$\rmS_1$ is given by
\begin{eqnarray}
y_{S_1}(3) &=&
\sqrt{E}hx_R+n_{S_1}(3)\\
&=&\eta_1h^2Ex_1+\eta_2hfEx_2+\eta_1h\sqrt{E}n_R(1)+\eta_2h\sqrt{E}n_R(2)+n_{S_1}(3).
\end{eqnarray}
As in the ANC protocol, $\rmS_1$ can completely remove $x_1$ and
obtain a new signal $\tilde{y}_{S_1}(3)$ given by
\begin{equation}
\tilde{y}_{S_1}(3)=\eta_2hfEx_2+\eta_1h\sqrt{E}n_R(1)+\eta_2h\sqrt{E}n_R(2)+n_{S_1}(3).
\end{equation}
Then $\rmS_1$ combines $\tilde{y}_{S_1}(3)$ with $y_{S_1}(2)$ by
maximum ratio combining and the instantaneous SNR
$\gamma_1^{\TDBC}$ of the combined signal is
\begin{eqnarray}
\gamma_1^{\TDBC} =
E|g|^2+\frac{E^2\eta_2^2|hf|^2}{E|h|^2(\eta_1^2+\eta_2^2)+1}
\approx E|g|^2+\frac{(1-\xi)E|hf|^2}{|f|^2(\xi
+1)+(1-\xi)|h|^2},\label{eqn:snr1tdbc}
\end{eqnarray}
where the approximation is due to
$\eta_1\approx\sqrt{\xi/(E|h|^2)}$ and
$\eta_2\approx\sqrt{(1-\xi)/(E|f|^2)}$. Similarly, the
instantaneous SNR $\gamma_2^{\TDBC}$ at $\rmS_2$ is
\begin{eqnarray}
\gamma_2^{\TDBC} =
E|g|^2+\frac{E^2\eta_1^2|hf|^2}{E|f|^2(\eta_1^2+\eta_2^2)+1}
&\approx& E|g|^2+\frac{\xi E|hf|^2}{|h|^2(2-\xi)+\xi
|f|^2}.\label{eqn:snr2tdbc}
\end{eqnarray}


\section{Outage Probability and Finite-SNR Diversity-Multiplexing
Tradeoff}\label{sec:prob}

In this section, we derive lower bounds of the outage
probabilities of the ANC and TDBC protocols. They are very tight
in the whole SNR range, irrespective of the values of channel
variances. Furthermore, based on those lower bounds, we derive the
finite-SNR diversity-multiplexing tradeoffs of the ANC and TDBC
protocols.

\subsection{Analog Network Coding}\label{sec:ANCout}

When the ANC protocol is used, it follows from (\ref{eqn:snr1anc})
and (\ref{eqn:snr2anc}) that the mutual information at $\rmS_1$
and $\rmS_2$ is given by
\begin{eqnarray}\label{eqn:ancr1r2}
I_1^\ANC =
\frac{1}{2}\log\left(1+\gamma_1^{\ANC}\right),&&I_2^\ANC =
\frac{1}{2}\log\left(1+\gamma_2^{\ANC}\right).
\end{eqnarray}
Note that the pre-log factor $1/2$ is because the information
exchange between the two sources takes two time slots
\cite{rankov}. Assume that the target rate of the whole
bidirectional cooperative network is $R$. Since the two sources in
this network are equivalent terminals, it is fair to set the
target rate of each source as $R/2$. Furthermore, in a
bidirectional cooperative network, the two sources are not only
transmitters but also receivers. Due to this reason, a
bidirectional cooperative network can be seen as a multiuser
system. It is well known that a multiuser system is in outage when
any user is in outage \cite{tse2}, \cite{tse3}. Therefore, the ANC
protocol is in outage when either $I_1^{\ANC}$ or $I_2^\ANC$ is
smaller than the target rate $R/2$, i.e.
\begin{equation}\label{eqn:ancout}
P_{\ANC}^{\outage}(R) = \Pr\left(I_1^\ANC<\frac{R}{2} \ \
{\mathrm{or}} \ \ I_2^\ANC<\frac{R}{2}\right).
\end{equation}

The exact expression of the outage probability is very hard to
derive even after we approximate the instantaneous SNRs as in
(\ref{eqn:snr1anc}) and (\ref{eqn:snr2anc}). However, it is well
known that the harmonic mean of two positive numbers can be
upper-bounded by the minimum of those two numbers as follows
\cite{anghel1}:
\begin{equation}\label{eqn:hormonic}
\frac{xy}{x+y}<\min(x,y).
\end{equation}
In order to make the analysis feasible, we use this method to
upper-bound the instantaneous SNRs and derive a lower bound of the
outage probability in the following lemma.\footnote{In fact, the
harmonic mean can also be lower-bounded as $1/2\min(x,y)\leq
xy/(x+y)$ \cite{anghel1}. Using this fact and the techniques
developed in Appendix A, we can find an upper bound of the outage
probability as well. However, this upper bound is not as tight as
the lower bound given in Lemma \ref{thm:ancoutage}, and hence, it
is not presented in this paper.}



\begin{lemma}\label{thm:ancoutage}
The outage probability of the ANC protocol can be lower-bounded as
follows:
\begin{eqnarray}\label{eqn:outANC}
P_{\ANC}^{\outage}(R)>\tilde{P}_{\ANC}^{\outage}(R)= 1- \exp\left(
-\frac{2\left(\Omega_h+\Omega_f\right)\left(2^{R}-1\right)}{E\Omega_h\Omega_f}\right).
\end{eqnarray}
\end{lemma}
\begin{proof}
See Appendix A.
\end{proof}


We notice that, irrespective of the values of channel variances,
the lower bound $\tilde{P}_\ANC^\outage(R)$ is extremely tight in
the whole SNR range as shown in Figs.\
\ref{fig:outagedis}--\ref{fig:outagpower2}. Due to this reason, we
use this lower bound to find the finite-SNR diversity-multiplexing
tradeoff of the ANC protocol as in \cite{nara} and it is given in
the follow theorem.
\begin{theorem}\label{cor:anccurve}
The finite-SNR diversity-multiplexing tradeoff curve $d^\ANC(r,E)$
of the ANC protocol is given by
\begin{eqnarray}\label{eqn:anccurvefinite}
d^\ANC(r,E)=\frac{2(\Omega_h+\Omega_f)\left(1+rE(1+E)^{r-1}-(1+E)^r\right)}
{E\Omega_h\Omega_f\left(1-\exp\left(\frac{2(\Omega_h+\Omega_f)}{E\Omega_h\Omega_f}\left((1+E)^r-1\right)\right)\right)}
,&&0\leq r\leq 1.
\end{eqnarray}
\end{theorem}
\begin{proof}
By definition, the finite-SNR diversity-multiplexing tradeoff
curve $d^\ANC(r,E)$ is given by \cite{nara}
\begin{equation}\label{eqn:anctradeoffdef}
d^\ANC(r,E) = -\frac{\partial \ln
\tilde{P}^\outage_\ANC(r\log(1+aE))}{\partial \ln E},
\end{equation}
where $a$ is a constant and called array gain. Since every
terminal has only one antenna, we set $a=1$ according to its
definition in \cite{nara}. By substituting (\ref{eqn:outANC}) into
(\ref{eqn:anctradeoffdef}), one can easily obtain
(\ref{eqn:anccurvefinite}).
\end{proof}

An interesting special case of the trade-off curve $d^\ANC(r,E)$
is when the SNR goes to infinity. This actually corresponds to the
definition of infinite-SNR diversity-multiplexing tradeoff given
in \cite{zheng}. For this case, the tradeoff curve of the ANC
protocol becomes
\begin{equation}\label{eqn:anccurveinfinite}
\lim_{E\rightarrow\infty}d^\ANC(r,E) = 1-r, \ \ \ \ 0\leq r\leq 1.
\end{equation}
Based on (\ref{eqn:anccurveinfinite}), we see that the maximum
multiplexing gain of the ANC protocol is one. This is because,
although each traffic flow in the ANC protocol takes two time
slots to complete the transmission, there are two concurrent
traffic flows and they are supported by the same physical
channels. \footnote{Actually, we conjecture that the highest
multiplexing gain any bidirectional cooperative network can
achieve is also one. It is certainly necessary to find the optimum
diversity-multiplexing tradeoff of a bidirectional cooperative
network as the authors did for the unidirectional cooperative
networks in \cite{azarian}, \cite{yang}. However, this is beyond
the scope of this paper where we focus on deriving and comparing
the tradeoffs of two specific protocols.} Note that, in a
conventional unidirectional cooperative network, there is only one
traffic flow in the network, and hence, the maximum multiplexing
gain of such network is just $1/2$. As a result, the ANC protocol
indeed has much higher bandwidth efficiency than the conventional
unidirectional cooperative networks. However, the maximum
diversity gain of the ANC protocol is just one. This is because
the ANC protocol let the two sources transmit simultaneously at
the first time slot. As a result, the sources can not receive the
signals from the direct channel due to the half-duplex constraint.
In order to achieve higher diversity gain, we can assign one time
slot to each source and let them transmit separately as in the
TDBC protocol discussed in the next subsection. As we will show,
however, this leads to some loss of multiplexing gain.

\subsection{Time Division Broadcasting}

When a bidirectional cooperative network implements the TDBC
protocol, the mutual information at $\rmS_1$ and $\rmS_2$ is given
by
\begin{eqnarray}\label{eqn:tdbci1i2}
I_1^\TDBC =
\frac{1}{3}\log\left(1+\gamma_1^\TDBC\right),&&I_2^\TDBC =
\frac{1}{3}\log\left(1+\gamma_2^\TDBC\right),
\end{eqnarray}
where the instantaneous SNRs are given by (\ref{eqn:snr1tdbc}) and
(\ref{eqn:snr2tdbc}), respectively. The pre-log factor $1/3$ is
because each traffic flow takes three time slots in the TDBC
protocol. As for the ANC protocol, the TDBC protocol is in outage
when either $I_1^\TDBC$ or $I_2^\TDBC$ is smaller than the target
rate $R/2$, i.e.
\begin{equation}\label{eqn:tdbcout}
P_{\TDBC}^{\outage}(R) = \Pr\left(I_1^\TDBC<\frac{R}{2} \ \
{\mathrm{or}} \ \ I_2^\TDBC<\frac{R}{2}\right).
\end{equation}
It is very hard to obtain the exact expression of
$P_{\TDBC}^{\outage}(R)$ even after using the approximations in
(\ref{eqn:snr1tdbc}) and (\ref{eqn:snr2tdbc}). In order to make
the problem tractable, we still use (\ref{eqn:hormonic}) to
upper-bound the instantaneous SNRs.
Furthermore, we let $\xi=1/2$ to simplify the derivation in this
subsection. This actually means that $\rmR$ uses half of its power
to transmit to $\rmS_1$ and uses the other half to $\rmS_2$. In
Subsection \ref{sec:TDBCpower}, we will consider how $\rmR$ can
optimally allocate its transmission power for $\rmS_1$ and
$\rmS_2$ by choosing an optimum value of $\xi$. A lower bound of
the outage probability of the TDBC protocol is given in the
following lemma.
\begin{lemma}\label{thm:outTDBC}
When $\xi=1/2$, the outage probability $P_\TDBC^\outage(R)$ of the
TDBC protocol can be lower-bounded as follows:
\begin{eqnarray}\label{eqn:outTDBC}
P_\TDBC^\outage(R)>\tilde{P}_\TDBC^\outage(R)&=&
1-\frac{1}{3\Omega_f\Omega_g+3\Omega_h\Omega_g-\Omega_h\Omega_f}\nonumber\\
&&\times\left(
3\Omega_g\left(\Omega_f+\Omega_h\right)\exp\left(-\frac{2^{1.5R}-1}{E\Omega_g}\right)\right.\nonumber\\
&&\left.-\Omega_h\Omega_f\exp\left(-\frac{3(2^{1.5R}-1)(\Omega_h+\Omega_f)}{E\Omega_h\Omega_f}\right)
\right).
\end{eqnarray}
\end{lemma}
\begin{proof}
See Appendix B.
\end{proof}

Figs.\ \ref{fig:outagedis}--\ref{fig:outagpower2} demonstrate that
this lower bound $\tilde{P}_\TDBC^\outage(R)$ is extremely tight
in the whole SNR range, irrespective of the values of channel
variances. Therefore, we use this lower bound to find the
finite-SNR diversity-multiplexing tradeoff of the TDBC protocol
and it is given in the follow theorem.
\begin{theorem}\label{cor:tdbccurve}
The finite-SNR diversity-multiplexing tradeoff curve
$d^\TDBC(r,E)$ of the TDBC protocol is given by
\begin{eqnarray}
d^\TDBC(r,E) = \frac{d_1(r,E)}{d_2(r,E)},&& 0\leq r\leq
\frac{2}{3},
\end{eqnarray}
where
\begin{eqnarray}
d_1(r,E) &=&
3(\Omega_h+\Omega_f)(\lambda(3rE-2E-2)+2E+2)\nonumber\\
&&\times\left(\exp\left(-\frac{3(\Omega_h+\Omega_f)}{E\Omega_h\Omega_f}(\lambda-1)\right)
-\exp\left(-\frac{\lambda-1}{E\Omega_g}\right)
\right),\\
d_2(r,E)&=&2E(E+1)\left(3\Omega_g(\Omega_h+\Omega_f)\left( 1-
\exp\left(-\frac{\lambda-1}{E\Omega_g}\right)\right)\right.\nonumber\\
&&\left.-\Omega_h\Omega_f\left(1-
\exp\left(-\frac{3(\Omega_h+\Omega_f)}{E\Omega_h\Omega_f}(\lambda-1)\right)\right)\right),\\
\lambda&=&(1+E)^{\frac{3r}{2}}.
\end{eqnarray}
\end{theorem}
\begin{proof}
The proof is essentially the same as that of Theorem
\ref{cor:anccurve}, except that we use the lower bound
$\tilde{P}_\TDBC^\outage(R)$ of the TDBC protocol instead of
$\tilde{P}_\ANC^\outage(R)$.
\end{proof}

When the SNR goes to infinity, the trade-off curve $d^\TDBC(r,E)$
becomes
\begin{eqnarray}\label{eqn:tdbccurveinfinite}
\lim_{E\rightarrow\infty} d^\TDBC(r,E) = 2 - 3r, &&0\leq r\leq
\frac{2}{3}.
\end{eqnarray}
As a result, the maximum multiplexing gain of the TDBC protocol is
$2/3$, which is smaller than that of the ANC protocol. This is
because, although the TDBC protocol supports two traffic flows
concurrently as the ANC protocol, each traffic flow takes three
time slots to complete the transmission as opposed to two time
slots in the ANC protocol. Due to this reason, the bandwidth
efficiency of the TDBC protocol is not as high as that of the ANC
protocol. However, note that the maximum multiplexing gain of the
TDBC protocol is still larger than that of the conventional
unidirectional cooperative networks. On the other hand, the
maximum diversity gain of the TDBC protocol is indeed two as shown
in \cite{kim1}, \cite{kim2}, and it is much larger than that of
the ANC protocol. Moreover, based on (\ref{eqn:anccurveinfinite})
and (\ref{eqn:tdbccurveinfinite}), those two protocols achieve the
same diversity gain at $r=1/2$ when SNR goes to infinity.

Many previous publications on bidirectional cooperative networks
did not consider the direct channel between the two sources
\cite{shengli1}, \cite{katti2}, \cite{shengli2}, and
\cite{liu}--\cite{lee2}. As a result, the maximum diversity gain
of the TDBC protocol is just one which is the same as that of the
ANC protocol. Since the maximum multiplexing gain of the TDBC
protocol is never larger than that of the ANC protocol, the ANC
protocol always outperforms the TDBC protocol in terms of
diversity-multiplexing tradeoff. When the direct channel exists,
however, the ANC protocol no longer outperforms the TDBC protocol
for all the time. As shown in Fig.\ \ref{fig:tradeoff}, those two
protocols achieve the same diversity gain when the multiplexing
gain is approximately $1/2$. When the multiplexing gain becomes
smaller, the TDBC protocol outperforms the ANC protocol as it
achieves a higher diversity gain; while, when the multiplexing
gain becomes bigger, the ANC protocol outperforms. This implies
that the ANC protocol can transmit information more efficiently,
but the TDBC protocol can transmit information more reliably.
Therefore, one may alternatively use those two protocols depending
on the specific task of a bidirectional cooperative network.

Certainly, it is desirable to find when the ANC and TDBC protocols
have the same diversity gain for a fixed $E$ by solving the
equation $d^\ANC(r,E)=d^\TDBC(r,E)$. Let $Q(E)$ denote the
solution to this equation. However, the exact expression of $Q(E)$
can not be given in closed form. In the following corollary, we
present a very accurate approximation to $Q(E)$.

\begin{cor}\label{cor:crosspoint}
When $E$ is fixed, the solution $Q(E)$ to the equation
$d^\ANC(r,E)=d^\TDBC(r,E)$ can be approximated by
\begin{eqnarray}
Q(E)\approx\tilde{Q}(E) =
\frac{1}{2}-\frac{\nu_1-\mu_1}{\nu_2-\mu_2}.
\end{eqnarray}
The coefficients $\mu_1$, $\mu_2$, $\nu_1$, and $\nu_2$ are given
as follows:
\begin{eqnarray}
\mu_1 &=& \frac{2\left( \Omega_h+\Omega_f \right)  \left( \,{\frac
{E}{2\sqrt
{1+E}}}+1-\sqrt {1+E} \right)}{ {E}{\Omega_h}{\Omega_f} \left( 1-c_1 \right)},\label{eqn:mu1}\\
\mu_2 &=&
\frac{2(\Omega_h+\Omega_f)}{E\Omega_h\Omega_f(1-c_1)}\left(\frac{E\ln(1+E)+2E}{2\sqrt{1+E}}-\sqrt{1+E}\ln(1+E)\right.\\
&&\left.+\frac{c_1(\Omega_h+\Omega_f)(E-2\sqrt{1+E}+2)\ln(1+E)}{E\Omega_h\Omega_f(c_1-1)}
\right),\label{eqn:mu2}\\
\nu_1&=&\frac{3(c_3-c_2)(\Omega_h+\Omega_f)\left(4(1+E)^{\frac{1}{4}}-E-4\right)}{E(1+E)^{\frac{1}{4}}
(3\Omega_g(\Omega_h+\Omega_f)(1-c_2)-\Omega_h\Omega_f(1-c_3))},\label{eqn:nu1}
\end{eqnarray}
\begin{eqnarray}
\nu_2 &=& \frac{3(\Omega_h+\Omega_f)}{E(1+E)^{\frac{1}{4}}}\left((c_3-c_2)\left(E-\frac{1}{4}(E+4)\ln(1+E)\right)\right.\nonumber\\
&&\left.+\frac{\ln(1+E)}{2E}\left(\frac{c_2}{\Omega_g}
-\frac{3c_3(\Omega_h+\Omega_f)}{\Omega_h\Omega_f}\right)(1+E)^{\frac{3}{4}}\left(2(1+E)^{\frac{1}{4}}
-\frac{E}{2}-2 \right) \right)\nonumber\\
&&+\frac{9(c_2-c_3)^2(\Omega_h+\Omega_f)^2\left(E+4-4(1+E)^{\frac{1}{4}}\right)(1+E)^{\frac{3}{4}}\ln(1+E)}
{4E^2(1+E)^{\frac{1}{4}}(3\Omega_g(\Omega_h+\Omega_f)(c_2-1)+\Omega_h\Omega_f(1-c_3))},\label{eqn:nu2}
\end{eqnarray}
where $c_1$, $c_2$, $c_3$ are constants and they are given by
\begin{eqnarray}
c_1 &=&\exp\left(\frac { 2\left( \Omega_h+\Omega_f \right) \left(
\sqrt {1+E}-1 \right) }{E\Omega_h\Omega_f}\right)\\
c_2 &=&\exp\left(\frac{1-(1+E)^{\frac{3}{4}}}{E\Omega_g}\right)\\
c_3
&=&\exp\left(\frac{3(\Omega_h+\Omega_f)(1-(1+E)^{\frac{3}{4}})}{E\Omega_g}\right).
\end{eqnarray}
\end{cor}
\begin{proof}
In Fig.\ \ref{fig:tradeoff}, we notice that both $d^\ANC(r,E)$ and
$d^\TDBC(r,E)$ can be accurately approximated by linear functions.
Furthermore, the solution $Q(E)$ to this equation
$d^\ANC(r,E)=d^\TDBC(r,E)$ is very close to $r=1/2$, irrespective
of the value of $E$. Therefore, we approximate $d^\ANC(r,E)$ and
$d^\TDBC(r,E)$ by Taylor expansion as follows:
\begin{eqnarray}\label{eqn:taylor}
d^\ANC(r,E) \approx \mu_1+\mu_2\left(r-\frac{1}{2}\right),&&
d^\TDBC(r,E) \approx \nu_1+\nu_2\left(r-\frac{1}{2}\right),
\end{eqnarray}
where the coefficients $\mu_1$, $\mu_2$, $\nu_1$, and $\nu_2$ are
given in (\ref{eqn:mu1})--(\ref{eqn:nu2}). By using
(\ref{eqn:taylor}), one can easily obtain $\tilde{Q}(E)$.
\end{proof}

Fig.\ \ref{fig:tradeoff2} demonstrates that $\tilde{Q}(E)$ is a
very tight approximation of $Q(E)$. When the SNR goes to infinity,
we can show that
\begin{equation}
\lim_{E\rightarrow\infty}\tilde{Q}(E) = \frac{1}{2}.
\end{equation}
This coincides with our conclusion based on
(\ref{eqn:anccurveinfinite}) and (\ref{eqn:tdbccurveinfinite}).
Furthermore, it is not hard to analytically show that
$\tilde{Q}(E)$ is always smaller than or equal to $1/2$ for any
$E$, which implies that the ANC protocol achieves a higher
diversity gain than the TDBC protocol as long as the multiplexing
gain is lager than $1/2$.

Note that our outage probability lower bounds and finite-SNR
diversity-multiplexing tradeoffs of the ANC and TDBC protocols are
based on the assumption that there is a direct channel between the
two sources. When such direct channel does not exist as in
\cite{shengli1}, \cite{katti2}, \cite{shengli2}, and
\cite{liu}--\cite{haoy}, however, out results can be easily
extended to this special case by letting $\Omega_g=0$. For the ANC
protocol, the lower bound $\tilde{P}_\ANC^\outage(R)$ and
diversity-multiplexing tradeoff $d^\ANC(r,E)$ do not change. For
the TDBC protocol, the lower bound of the outage probability and
the diversity-multiplexing tradeoff now become
\begin{eqnarray}
\tilde{P}_\TDBC^\outage(R)&=&
1-\exp\left(-\frac{3(2^{1.5R}-1)(\Omega_h+\Omega_f)}{E\Omega_h\Omega_f}
\right),
\end{eqnarray}
\begin{equation}
d^\TDBC(r,E)=\frac{3(\Omega_h+\Omega_f)(\lambda(3rE-2E-2)+2E+2)
\exp\left(-\frac{3(\Omega_h+\Omega_f)}{E\Omega_h\Omega_f}(\lambda-1)\right)
}{2E(E+1)\Omega_h\Omega_f\left(
\exp\left(-\frac{3(\Omega_h+\Omega_f)}{E\Omega_h\Omega_f}(\lambda-1)\right)-1\right)}.
\end{equation}
Furthermore, when the SNR goes to infinity, the trade-off curve of
the TDBC protocol becomes
\begin{eqnarray}\label{eqn:tdbccurveinfinite2}
\lim_{E\rightarrow\infty} d^\TDBC(r,E) = 1 - \frac{3}{2}r, &&0\leq
r\leq \frac{2}{3}.
\end{eqnarray}
By comparing (\ref{eqn:tdbccurveinfinite2}) and
(\ref{eqn:anccurveinfinite}), one can see that the ANC protocol
indeed always outperforms the TDBC protocol in terms of
diversity-multiplexing tradeoff when the direct channel does not
exist. This coincides with the conclusion which can be drawn from
\cite{shengli1}, \cite{katti2}, and \cite{haoy}.


\section{Optimum Power Allocation}\label{sec:power}

In this section, we first propose an optimum power allocation
scheme for the ANC protocol. This scheme can simultaneously
minimize the outage probability and maximize the total mutual
information of the ANC protocol. Secondly, we develop an optimum
method for the TDBC protocol to combine the received signals at
the relay. This method also minimizes the outage probability and
maximizes the total mutual information of the TDBC protocol at the
same time.

\subsection{Analog Network Coding}

In Subsection \ref{sec:ANCout}, we derive a lower bound of the
outage probability of the ANC protocol when every terminal has the
same transmission power $E$. Since every terminal knows the values
of $h$ and $f$, it is more desirable to allocate the transmission
power according to channel conditions in order to maximize the
performance. Such power allocation problem has not been
investigated yet in previous publications. We now assume that the
transmission powers of $\rmS_1$, $\rmS_2$, and $\rmR$ are $E_1$,
$E_2$, and $E_r$, respectively. Consequently, the instantaneous
SNRs at $\rmS_1$ and $\rmS_2$ should be rewritten as
\begin{eqnarray}\label{eqn:ancsnr12gen}
\gamma_1^{\ANC}
\approx\frac{E_rE_2|hf|^2}{(E_r+E_1)|h|^2+E_2|f|^2},&&
\gamma_2^{\ANC} \approx
\frac{E_rE_1|hf|^2}{(E_r+E_2)|f|^2+E_1|h|^2}.
\end{eqnarray}

When it comes to optimum power allocation, two optimization goals
are usually considered: minimization of the outage probability and
maximization of the total mutual information.\footnote{For a
single user system whose outage probability is formulated by
$\Pr(I<R)$, it is not hard to see that the optimum power
allocation scheme that minimizes the outage probability also
maximizes the mutual information. For the bidirectional
cooperative network considered in this paper, however, its outage
probability is given in the form $\Pr(I_1<R/2\ {\mathrm{or}}\
I_2<R/2)$ and its total mutual information is $I_1+I_2$.
Therefore, it is not easy to see if there exists an optimum power
allocation scheme which can minimize the outage probability and
maximize the total mutual information at the same time.} We first
use the outage probability as the metric to optimally allocate the
power. That is, the optimization problem is formulated by
\begin{eqnarray}\label{eqn:ancopt1}
(E_1,E_2,E_r) = \arg\min_{E_1,E_2,E_r} P_\ANC^\outage(R),
&&{\mathrm{subject \ to}}\ \ \ E_1+E_2+E_r=3E.
\end{eqnarray}
It follows from (\ref{eqn:ancout}) that the optimization problem
in (\ref{eqn:ancopt1}) is equivalent to the following one
\begin{eqnarray}\label{eqn:ancminmax}
(E_1,E_2,E_r) = \arg\max_{E_1,E_2,E_r}\min
(\gamma_1^\ANC,\gamma_2^\ANC),&&{\mathrm{subject \ to}}\ \ \
E_1+E_2+E_r=3E.
\end{eqnarray}
The minimax problem in (\ref{eqn:ancminmax}) is solved in the
following theorem.
\begin{theorem}\label{thm:ancpower1}
When $E_1+E_2+E_r=3E$, the optimum power allocation that minimizes
the outage probability $P_\ANC^\outage(R)$ of the ANC protocol is
given by
\begin{eqnarray}
E_r &=& \frac{3}{2}E,\label{eqn:ancer}\\
E_1 &=& \frac{3|f|}{2(|h|+|f|)}E,\label{eqn:ance1}\\
E_2 &=& \frac{3|h|}{2(|h|+|f|)}E\label{eqn:ance2}.
\end{eqnarray}
\end{theorem}
\begin{proof}
See Appendix C.
\end{proof}

Secondly, we consider the optimum power allocation scheme that
maximizes the total mutual information $I^\ANC=I_1^\ANC+I_2^\ANC$
of the ANC protocol. Therefore, the optimization problem is
formulated as
\begin{eqnarray}\label{eqn:ancoptth1}
(E_1,E_2,E_r) = \arg\max_{E_1,E_2,E_r} I^\ANC, &&{\mathrm{subject
\ to}}\ \ \ E_1+E_2+E_r=3E.
\end{eqnarray}
This optimization problem is solved in the following lemma.
\begin{lemma}
When $E_1+E_2+E_r=3E$, the optimum power allocation scheme that
maximizes the total mutual information $I^\ANC$ of the ANC
protocol is given by (\ref{eqn:ancer})--(\ref{eqn:ance2}).
\end{lemma}
\begin{proof}
Let $E_1=3\alpha E$, $E_2=3\beta E$, and $E_r =
3(1-\alpha-\beta)E$, where $0\leq \alpha\leq 1$, $0\leq \beta\leq
1$, and $\alpha+\beta\leq 1$. With those constraints, the solution
to $\partial I^\ANC/\partial \alpha=0$ and $\partial
I^\ANC/\partial\beta=0$ is given by
\begin{eqnarray}
\alpha = \frac{|f|}{2(|h|+|f|)},&& \beta = \frac{|h|}{2(|h|+|f|)}.
\end{eqnarray}
As a result, the optimum power allocation scheme is given by
(\ref{eqn:ancer})--(\ref{eqn:ance2}).
\end{proof}

Interestingly, the optimum power allocation schemes based on
outage probability and total mutual information are exactly the
same. This is due to the special structures of the instantaneous
SNRs $\gamma_1^\ANC$ and $\gamma_2^\ANC$ given in
(\ref{eqn:ancsnr12gen}). As a result, we obtain an optimum power
allocation scheme that simultaneously minimizes the outage
probability and maximizes the total mutual information of the ANC
protocol.

\subsection{Time Division Broadcasting}\label{sec:TDBCpower}

Unlike the ANC protocol, it is very hard to find an optimum power
allocation scheme for the TDBC protocol no matter using the outage
probability or the total mutual information as the criteria. For
example, if we intend to minimize the outage probability, the
optimization problem is a generalized fractional programming
problem. Such problem can only be solved numerically for most
cases \cite{barros}, \cite{crouzeix}.\footnote{In fact, for the
ANC protocol, the optimization problem (\ref{eqn:ancminmax}) is
also a generalized fractional programming problem. We obtain a
closed form solution to (\ref{eqn:ancminmax}) only because
$\gamma_1^\ANC$ and $\gamma_2^\ANC$ have very special structures.}
However, the optimum value of $\xi$ can be analytically found in
closed form and it can greatly improve the performance. We first
try to find the optimum value of $\xi$ that minimizes the outage
probability of the TDBC protocol. That is, we consider the
following optimization problem
\begin{eqnarray}\label{eqn:tdbcminmax1}
\xi = \arg\min_{\xi}P_\TDBC^\outage(R),&&{\mathrm{subject \ to}} \ \ \ 0\leq\xi\leq 1.
\end{eqnarray}
It follows from (\ref{eqn:tdbcout}) that the optimization problem
in (\ref{eqn:tdbcminmax1}) is equivalent to the following one
\begin{eqnarray}\label{eqn:tdbcminmax2}
\xi = \arg\mathop{\max}_{\xi}\min
(\gamma_1^\TDBC,\gamma_2^\TDBC),&&{\mathrm{subject \ to}} \ \ \ 0\leq\xi\leq 1.
\end{eqnarray}
The solution to (\ref{eqn:tdbcminmax2}) is given in a simple and
closed form in the following theorem.
\begin{theorem}\label{thm:tdbcout}
The optimum value of $\xi$ that minimizes the outage probability
$P_\TDBC^\outage(R)$ of the TDBC protocol is given by
\begin{equation}
\xi = \frac{|h|}{|h|+|f|}\label{eqn:xitdbc1}.
\end{equation}
\end{theorem}
\begin{proof}
See Appendix D.
\end{proof}

Note that the optimum value of $\xi$ given in Theorem
\ref{thm:tdbcout} is based on the assumption that every terminal
has the same transmission power $E$. Certainly, one can jointly
optimize the value of $\xi$ and the transmission power of every
terminal to minimize the outage probability, but this can only be
done numerically. Actually, by letting every terminal has the same
transmission power $E$ and letting $\xi$ equal to $|h|/(|h|+|f|)$
as in (\ref{eqn:xitdbc1}), the performance of the network is
almost the same as that of the network where the transmission
powers and the value of $\xi$ are jointly optimized by numerical
ways as shown in Figs.\ \ref{fig:outagpower1} and
\ref{fig:outagpower2}.

Secondly, we investigate the optimum value of $\xi$ that maximizes
the total mutual information $I^\TDBC=I_1^\TDBC+I_2^\TDBC$ of the
TDBC protocol, i.e.
\begin{eqnarray}\label{eqn:tdbcopt2}
\xi = \arg\max_{\xi} I^\TDBC,&& {\mathrm{subject \ to}} \ \ 0\leq
\xi\leq 1,
\end{eqnarray}
where $I_1^\TDBC$ and $I_2^\TDBC$ are given by
(\ref{eqn:tdbci1i2}). The solution to (\ref{eqn:tdbcopt2}) is
given the following lemma.
\begin{lemma}\label{lem:tdbcpower2}
The optimum value of $\xi$ that maximizes the total mutual
information $I^\TDBC$ of the TDBC protocol is given by
(\ref{eqn:xitdbc1}).
\end{lemma}
\begin{proof}
By solving the equation $\partial I^\TDBC/\partial \xi=0$, we can
easily obtained the solution given in (\ref{eqn:xitdbc1}).
\end{proof}

As for Theorem \ref{thm:tdbcout}, the optimum value of $\xi$ in
Lemma \ref{lem:tdbcpower2} is also based on the setting that every
terminal has the same transmission power $E$. The joint
optimization of the transmission powers and the value of $\xi$ to
maximize the total mutual information can only be accomplished by
numerical ways. In fact, by simply letting every terminal has the
same transmission power $E$ and letting $\xi$ equal to
(\ref{eqn:xitdbc1}), the total mutual information of the TDBC
protocol can be improved substantially as shown in Fig.\
\ref{fig:ratedis}.

Interestingly, we notice that the optimum values of $\xi$ in
Theorem \ref{thm:tdbcout} and Lemma \ref{lem:tdbcpower2} are
exactly the same, although they are based on two different
criteria. As a result, by letting $\xi$ equal to $|h|/(|h|+|f|)$,
we can minimize the outage probability and maximize the total
mutual information of the TDBC protocol at the same time.
Actually, the reason why we can find such $\xi$ is because every
terminal has the same transmission power. In general, when every
terminal has unequal transmission power, the solutions to
(\ref{eqn:tdbcminmax2}) and (\ref{eqn:tdbcopt2}) are different.
Furthermore, the joint optimization of the transmission powers and
the value of $\xi$ based on the outage probability criteria is
also generally different with that based on the total mutual
information criteria. Thus, one may use equal power allocation and
set $\xi$ as (\ref{eqn:xitdbc1}) in order to transmit information
efficiently and reliably at the same time.


\section{Numerical Results}\label{sec:numeric}

In this section, we present some numerical results to demonstrate
the performance of the ANC and TDBC protocols. We assume that all
three terminals are located in a straight line and $\rmR$ is
between $\rmS_1$ and $\rmS_2$. We fix the distance between
$\rmS_1$ and $\rmS_2$ as one and let $D_1$ denote the distance
between $\rmS_1$ and $\rmS_2$. Furthermore, we set the path loss
factor as four in order to model radio propagation in urban areas
\cite{rappaport}. As a result, the values of $\Omega_g$,
$\Omega_h$ and $\Omega_f$ equal to one, $D_1^{-4}$, and
$(1-D_1)^{-4}$, respectively.

In Fig.\ \ref{fig:tradeoff}, we compare the diversity-multiplexing
tradeoffs of the ANC and TDBC protocols. For each protocol, its
finite-SNR diversity-multiplexing tradeoff indeed converges to the
infinite-SNR case when $E$ goes to infinity. When $E$ is fixed,
the tradeoff curves of those two protocols have a cross point at
approximately $r=1/2$ as expected by Corollary
\ref{cor:crosspoint}. Fig.\ \ref{fig:tradeoff2} shows the values
of $Q(E)$, and hence, it demonstrates when the ANC and TDBC
protocols have the same diversity gain. One can see that our
approximate solution $\tilde{Q}(E)$ is extremely tight to the
exact one $Q(E)$ which is obtained by a numerical method.

In Fig.\ \ref{fig:outagedis}, we compare the exact outage
probabilities of the ANC and TDBC protocols with our lower bounds
given in (\ref{eqn:outANC}) and (\ref{eqn:outTDBC}). It can be
seen that the lower bounds are extremely tight for both protocols
even when we change the location of the relay and the value of
multiplexing gain. Furthermore, in Figs.\ \ref{fig:outagpower1}
and \ref{fig:outagpower2}, we see that our lower bounds are
constantly tight in the whole SNR range. In Fig.\
\ref{fig:outagpower1}, we let the multiplexing gain equal to zero.
As a result, the TDBC protocol achieves a higher diversity gain
than the ANC protocol, which coincides with our conclusion drawn
from Fig.\ \ref{fig:tradeoff}. On the other hand, we let the
multiplexing gain equal to $0.6$ in Fig.\ \ref{fig:outagpower2}
and show that the ANC protocol has a higher diversity gain for
this case.

In Figs.\ \ref{fig:outagpower1} and \ref{fig:outagpower2}, one can
also see that the optimum power allocation scheme of the ANC
protocol substantially reduces the outage probability even when
the relay is exactly in the middle of the two sources. For the
TDBC protocol, its outage probability is considerably reduced as
well when the proposed combing method given in (\ref{eqn:xitdbc1})
is implemented with equal power allocation. Moreover, by letting
$\xi$ equal to (\ref{eqn:xitdbc1}) and adopting equal power
allocation, we can achieve almost the same outage probability as
jointly optimizing the value of $\xi$ and the transmission powers
by numerical ways.

In Fig.\ \ref{fig:ratedis}, we show that the optimum power
allocation scheme can considerably increase the total mutual
information of the ANC protocol especially when $\rmR$ is close to
either $\rmS_1$ or $\rmS_2$. Moreover, the proposed optimum
combing method with equal power allocation can greatly increase
the total mutual information of the TDBC protocol as well, but the
improvement is not as considerable as that achieved by joint
optimization of $\xi$ and the transmission powers.

\section{Conclusion}\label{sec:conl}

This paper studies the ANC and TDBC protocols which are used to
achieve information exchange in bidirectional cooperative
networks. We derive lower bounds of the outage probabilities of
those two protocols. The lower bounds are extremely tight in the
whole SNR range irrespective of the values of channel variances.
Therefore, based on those lower bounds, we derive the finite-SNR
diversity-multiplexing tradeoffs of the ANC and TDBC protocols.
Furthermore, we propose an optimum power allocation scheme for the
ANC protocol. This scheme can simultaneously minimize the outage
probability and maximize the total mutual information of the ANC
protocol. For the TDBC protocol, we develop an optimum combing
method for the relay terminal under an equal power allocation
assumption. This method substantially reduces the outage
probability and increases the total mutual information as well.


\section*{Appendix A}
\renewcommand\theequation{A.\arabic{equation}}
\setcounter{equation}{0}
\begin{center}
Proof of Lemma \ref{thm:ancoutage}
\end{center}

Let $X=|h|^2$ and $Y=|f|^2$. Thus, $X$ and $Y$ are exponential
random variables with means $\Omega_h$ and $\Omega_f$,
respectively. By using the inequality (\ref{eqn:hormonic}), the
outage probability can be lower-bounded as follows:
\begin{eqnarray}
P_{\ANC}^{\outage}(R) &=& \Pr\left(I_1^\ANC<\frac{R}{2} \ \ {\mathrm{or}} \ \ I_2^\ANC<\frac{R}{2}\right)\\
&=&1-\Pr\left(I_1^\ANC>\frac{R}{2}, I_2^\ANC>\frac{R}{2}\right)\\
&=&1-\Pr\left(\gamma_1^\ANC>2^{R}-1, \gamma_1^\ANC<\gamma_2^\ANC
\right)\nonumber\\
&&-\Pr\left(\gamma_2^\ANC>2^{R}-1,
\gamma_2^\ANC<\gamma_1^\ANC \right)\\
&>&1-\Pr\left(\frac{E}{2}\min(2X,Y)>2^{R}-1,\min(2X,Y)<\min(2Y,X)\right)\nonumber\\
&&-\Pr\left(\frac{E}{2}\min(2Y,X)>2^{R}-1,\min(2Y,X)<\min(2X,Y)\right).\label{eqn:outANC1}
\end{eqnarray}
The first probability in (\ref{eqn:outANC1}) can be evaluated in
the following way
\begin{eqnarray}
&&\Pr\left(\frac{E}{2}\min(2X,Y)>2^{R}-1,\min(2X,Y)<\min(2Y,X)\right)\nonumber\\
&=&\Pr\left(Y>\frac{2\left(2^{R}-1\right)}{E},\min(2Y,X)>Y,2X>Y\right)\\
&=&\Pr\left(Y>\frac{2\left(2^{R}-1\right)}{E},X>2Y\right)
+\Pr\left(Y>\frac{2\left(2^{R}-1\right)}{E},2Y>X>Y\right)\\
&=&\Pr\left(Y>\frac{2\left(2^{R}-1\right)}{E},X>Y\right)\\
&=&\frac{\Omega_h}{\Omega_h+\Omega_f}\exp\left(
-\frac{2\left(\Omega_h+\Omega_f\right)\left(2^{R}-1\right)}{E\Omega_h\Omega_f}\right),\label{eqn:outANCp1}
\end{eqnarray}
where the integrations involved in the last step can be solved by
\cite{int}. Similarly, the second probability in
(\ref{eqn:outANC1}) can be solved as follows:
\begin{eqnarray}\label{eqn:outANCp2}
&&\Pr\left(\frac{E}{2}\min(2Y,X)>2^{R}-1,\min(2Y,X)<\min(2X,Y)\right)\nonumber\\
&=&\frac{\Omega_f}{\Omega_h+\Omega_f}\exp\left(
-\frac{2\left(\Omega_h+\Omega_f\right)\left(2^{R}-1\right)}{E\Omega_h\Omega_f}\right).
\end{eqnarray}
By substituting (\ref{eqn:outANCp1}) and (\ref{eqn:outANCp2}) into
(\ref{eqn:outANC1}), we obtain the lower bound
$\tilde{P}_\ANC^\outage(R)$ of the outage probability
$P_\ANC^\outage(R)$.

\section*{Appendix B}
\renewcommand\theequation{B.\arabic{equation}}
\setcounter{equation}{0}
\begin{center}
Proof of Lemma \ref{thm:outTDBC}
\end{center}

Let $X=|h|^2$, $Y=|f|^2$, and $Z=|g|^2$. Thus, $X$, $Y$, and $Z$
are exponential random variables with means $\Omega_h$,
$\Omega_f$, and $\Omega_g$, respectively. By using the inequality
in (\ref{eqn:hormonic}), the outage probability is lower-bounded
by
\begin{eqnarray}
P_\TDBC^\outage(R) &=& \Pr\left(I_1^\TDBC<\frac{R}{2} \ \
\mathrm{or} \ \
I_2^\TDBC<\frac{R}{2}\right)\\
&=&1-\Pr\left(I_1^\TDBC>\frac{R}{2}, I_2^\TDBC>\frac{R}{2}\right)\\
&=&1-\Pr\left(\gamma_1^\TDBC>2^{1.5R}-1,
\gamma_1^\TDBC<\gamma_2^\TDBC
\right)\nonumber\\
&&-\Pr\left(\gamma_2^\TDBC>2^{1.5R}-1,
\gamma_2^\TDBC<\gamma_1^\TDBC \right)\\
&>&1-\Pr\left(\frac{E}{3}\min(3Y,X)+EZ>2^{1.5R}-1,\min(3Y,X)<\min(3X,Y)\right)\nonumber\\
&&-\Pr\left(\frac{E}{3}\min(3X,Y)+EZ>2^{1.5R}-1,\min(3X,Y)<\min(3Y,X)\right).\label{eqn:outTDBC1}
\end{eqnarray}

The first probability in (\ref{eqn:outTDBC1}) can be solved as
follows:
\begin{eqnarray}
&&\Pr\left(\frac{E}{3}\min(3Y,X)+EZ>2^{1.5R}-1,\min(3Y,X)<\min(3X,Y)\right)\nonumber\\
&=&\Pr\left(\frac{1}{3}X+Z>\frac{2^{1.5R}-1}{E},X<\min(3X,Y),X<3Y\right)\nonumber\\
&&+\Pr\left(aY+Z>\frac{2^{1.5R}-1}{E},3Y<\min(3X,Y),X>3Y\right)\\
&=&\Pr\left(\frac{1}{3}X+Z>\frac{2^{1.5R}-1}{E},X<3Y,3X<Y\right)\nonumber\\
&&+\Pr\left(\frac{1}{3}X+Z>\frac{2^{1.5R}-1}{E},X<Y,X<3Y,3X>Y\right)\\
&=&\Pr\left(\frac{1}{3}X+Z>\frac{2^{1.5R}-1}{E},X<\frac{Y}{3}\right)\nonumber\\
&&+\Pr\left(\frac{1}{3}X+Z>\frac{2^{1.5R}-1}{E},\frac{Y}{3}<X<Y\right)\\
&=&\Pr\left(\frac{1}{3}X+Z>\frac{2^{1.5R}-1}{E},X<Y\right)\\
&=&\frac{\Omega_f}{(3\Omega_f\Omega_g+3\Omega_h\Omega_g-\Omega_h\Omega_f)(\Omega_h+\Omega_f)}\left(
3\Omega_g\left(\Omega_f+\Omega_h\right)\exp\left(-\frac{2^{1.5R}-1}{E\Omega_g}\right)\right.\nonumber\\
&&\left.-\Omega_h\Omega_f\exp\left(-\frac{3(2^{1.5R}-1)(\Omega_h+\Omega_f)}{E\Omega_h\Omega_f}\right)
\right).\label{eqn:tdbcprob1}
\end{eqnarray}
where the integrations involved in the last step can be solved by
\cite{int}. Similarly, the second integration can be solved in the
following way
\begin{eqnarray}
&&\Pr\left(\frac{E}{3}\min(3X,Y)+EZ>2^{1.5R}-1,\min(3X,Y)<\min(3Y,X)\right)\nonumber\\
&=&\frac{\Omega_h}{(3\Omega_f\Omega_g+3\Omega_h\Omega_g-\Omega_h\Omega_f)(\Omega_h+\Omega_f)}\left(
3\Omega_g\left(\Omega_f+\Omega_h\right)\exp\left(-\frac{2^{1.5R}-1}{E\Omega_g}\right)\right.\nonumber\\
&&\left.-\Omega_h\Omega_f\exp\left(-\frac{3(2^{1.5R}-1)(\Omega_h+\Omega_f)}{E\Omega_h\Omega_f}\right)
\right).\label{eqn:tdbcprob2}
\end{eqnarray}
Based on (\ref{eqn:outTDBC1}), (\ref{eqn:tdbcprob1}), and
(\ref{eqn:tdbcprob2}), the outage probability of the TDBC protocol
is lower-bounded by $\tilde{P}_\TDBC^\outage(R)$.

\section*{Appendix C}
\renewcommand\theequation{C.\arabic{equation}}
\setcounter{equation}{0}
\begin{center}
Proof of Theorem \ref{thm:ancpower1}
\end{center}

We let $E_1=3\alpha E$, $E_2=3\beta E$, and $E_r =
3(1-\alpha-\beta)E$, where $0\leq \alpha\leq 1$, $0\leq \beta\leq
1$, and $\alpha+\beta\leq 1$. Consequently, $\gamma_1^\ANC$ in and
$\gamma_2^\ANC$ in (\ref{eqn:ancsnr12gen}) are approximated by
\begin{eqnarray}
\gamma_1^\ANC \approx
3|hf|^2E\frac{\beta(1-\alpha-\beta)}{|h|^2(1-\beta)+|f|^2\beta},&&
\gamma_2^\ANC \approx
3|hf|^2E\frac{\alpha(1-\alpha-\beta)}{|f|^2(1-\alpha)+|h|^2\alpha}.
\end{eqnarray}
We follow the method given in Section II.C of \cite{poor} to solve
the minimax problem in (\ref{eqn:ancminmax}). Thus, a new function
$M(\pi_0,\alpha,\beta)$ is defined as
\begin{equation}
M(\pi_0,\alpha,\beta) =
\pi_0\gamma_1^\ANC+(1-\pi_0)\gamma_2^\ANC,\ \ \ 0\leq \pi_0\leq 1.
\end{equation}
Let $\alpha(\pi_0)$ and $\beta(\pi_0)$ denote the values of
$\alpha$ and $\beta$, respectively, which maximize
$M(\pi_0,\alpha,\beta)$ for a fixed $\pi_0$. According to
\cite{poor}, the solution to (\ref{eqn:ancminmax}) must belong to
the set formed by $\alpha(\pi_0)$ and $\beta(\pi_0)$, i.e. it must
maximize $M(\pi_0,\alpha,\beta)$.

We try to find $\alpha(\pi_0)$ and $\beta(\pi_0)$ in the
following. When $\pi_0\neq 1/2$, it can be shown that the values
of $\alpha(\pi_0)$ and $\beta(\pi_0)$ are given by
\begin{eqnarray}\label{eqn:nothalf}
\left\{%
\begin{array}{ccl}
  \alpha(\pi_0) = 0, &
  \beta(\pi_0)=\frac{|h|}{|h|+|f|}, & {\mathrm{when}}\ 0\leq\pi_0<1/2\\
  \beta(\pi_0) = 0, &
  \alpha(\pi_0)=\frac{|f|}{|h|+|f|}, & {\mathrm{when}}\ 1/2<\pi_0\leq 1
\end{array}%
\right..
\end{eqnarray}
When either $\alpha=0$ or $\beta=0$, however, it follows from
(\ref{eqn:ancr1r2}) and (\ref{eqn:ancout}) that the outage
probability is one. This implies that the solution in
(\ref{eqn:nothalf}) can not be the solution to
(\ref{eqn:ancminmax}).

Consequently, the solution to (\ref{eqn:ancminmax}) can only be
found at $\pi_0=1/2$. For this case, we first find that
$\alpha(1/2)$ and $\beta(1/2)$ are not unique. Specifically,
$\beta(1/2)$ can be any real number between zero and one, and the
value of $\alpha(1/2)$ depends on $\beta(1/2)$ in the following
way
\begin{eqnarray}\label{eqn:alpha1}
\alpha(1/2) =
\frac{|f|\left(|hf|-(|f|^2\beta+|h|^2(1-\beta(1/2))\right)}{|h|(|f|^2-|h|^2)}.
\end{eqnarray}
On the other hand, it has been shown in \cite{poor} that the
solution to (\ref{eqn:ancminmax})
should make $\gamma_1^\ANC=\gamma_2^\ANC$. 
Based on this condition and (\ref{eqn:alpha1}), we should choose
$\alpha(1/2)$ and $\beta(1/2)$ as follows:
\begin{eqnarray}
\alpha(1/2) = \frac{|f|}{2(|h|+|f|)},&& \beta(1/2) =
\frac{|h|}{2(|h|+|f|)}.
\end{eqnarray}
Therefore, the optimum power allocation scheme that minimizes the
outage probability of the ANC protocol is given by
(\ref{eqn:ancer})--(\ref{eqn:ance2}).

\section*{Appendix D}
\renewcommand\theequation{D.\arabic{equation}}
\setcounter{equation}{0}
\begin{center}
Proof of Theorem \ref{thm:tdbcout}
\end{center}

In this proof, we rewrite $\gamma_1^\TDBC$ and $\gamma_2^\TDBC$ as
$\gamma_1^\TDBC(\xi)$ and $\gamma_2^\TDBC(\xi)$ in order to
emphasize their dependence on $\xi$. Furthermore, we construct a
new function $W(\pi_0,\xi)$ as follows:
\begin{equation}
W(\pi_0,\xi) =
\pi_0\gamma_1^\TDBC(\xi)+(1-\pi_0)\gamma_2^\TDBC(\xi),\ \ \ 0\leq
\pi_0\leq 1.
\end{equation}
Let $\xi(\pi_0)$ denote the value of $\xi$ that maximizes
$W(\pi_0,\xi)$ for a fixed $\pi_0$. Let
$V(\pi_0)=W(\pi_0,\xi(\pi_0))$ and let $\pi_L$ denote the solution
to $\arg\min_{0\leq \pi_0\leq 1} V(\pi_0)$. It has been shown that
the solution to (\ref{eqn:tdbcminmax2}) must be from the set
formed by $\xi(\pi_0)$ \cite{poor}. Furthermore, it follows from
Proposition II.C.1 in \cite{poor} that $\xi(\pi_L)$ is the
solution to (\ref{eqn:tdbcminmax2}) either
$\gamma_1^\TDBC(\xi(\pi_L))=\gamma_2^\TDBC(\xi(\pi_L))$,
$\pi_L=0$, or $\pi_L=1$.

As a result, in order to find the optimum value of $\xi$, we can
first investigate the equation
$\gamma_1^\TDBC(\xi)=\gamma_2^\TDBC(\xi)$. If a solution exits for
this equation, such solution must be the solution to
(\ref{eqn:tdbcminmax2}). Fortunately, the equation
$\gamma_1^\TDBC(\xi)=\gamma_2^\TDBC(\xi)$ always have a solution
under the assumption that every terminal has the same transmission
power $E$. Specifically, it can be easily shown that
$\gamma_1^\TDBC(\xi)$ is a strictly decreasing function of $\xi$;
while, $\gamma_2^\TDBC(\xi)$ is a strictly increasing function of
$\xi$. Furthermore, $\gamma_1^\TDBC(0)=\gamma_2^\TDBC(1)$ and
$\gamma_1^\TDBC(1)=\gamma_2^\TDBC(0)$. As result, the two
functions $\gamma_1^\TDBC(\xi)$ and $\gamma_2^\TDBC(\xi)$ always
have one and only one crossing point in the range $0\leq \xi \leq
1$, i.e. the equation $\gamma_1^\TDBC(\xi)=\gamma_2^\TDBC(\xi)$
always has a solution. Such solution is given by $|h|/(|h|+ |f|)$,
and hence, the solution to (\ref{eqn:tdbcminmax2}) is given by
(\ref{eqn:xitdbc1}).


\clearpage
\newpage
\listoffigures

\clearpage
\newpage

\begin{figure}
\begin{center}
\subpostscript{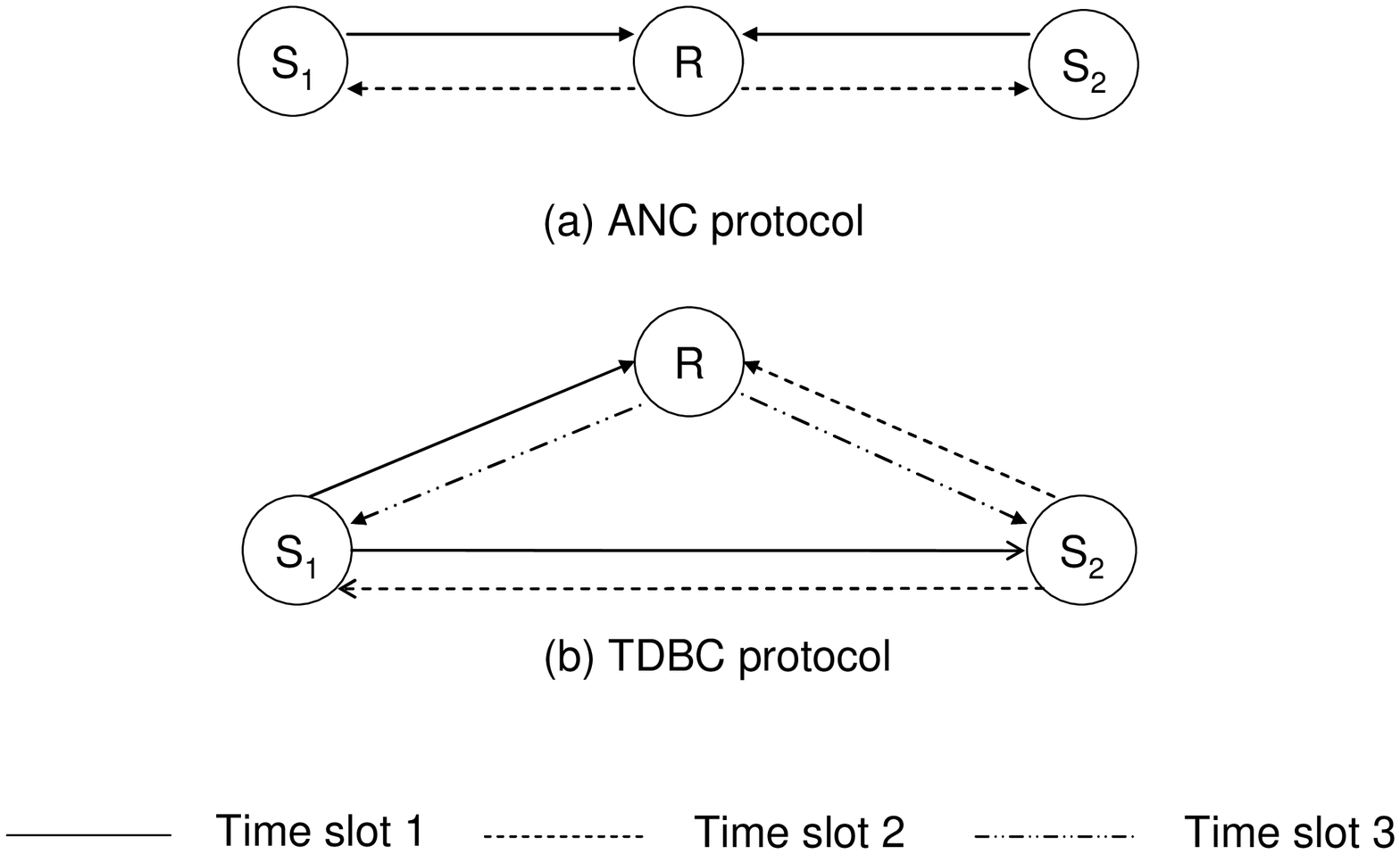}{0.9\textwidth}
\end{center}
\caption{System models of the ANC and TDBC protocols.}
\label{fig:model}
\end{figure}

\clearpage
\newpage
\begin{figure}
\begin{center}
\subpostscript{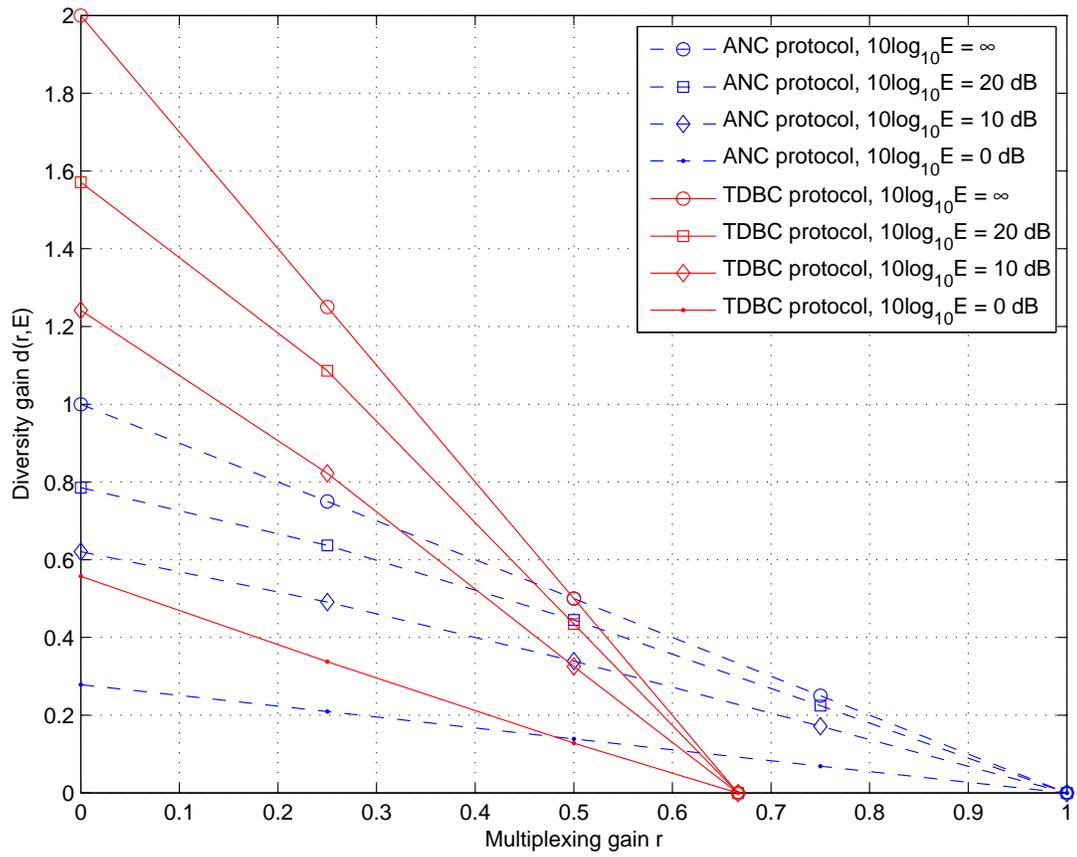}{0.9\textwidth}
\end{center}
\caption{Diversity-multiplexing tradeoffs of the ANC and TDBC
protocols, $D_1=0.5$.} \label{fig:tradeoff}
\end{figure}

\clearpage
\newpage
\begin{figure}
\begin{center}
\subpostscript{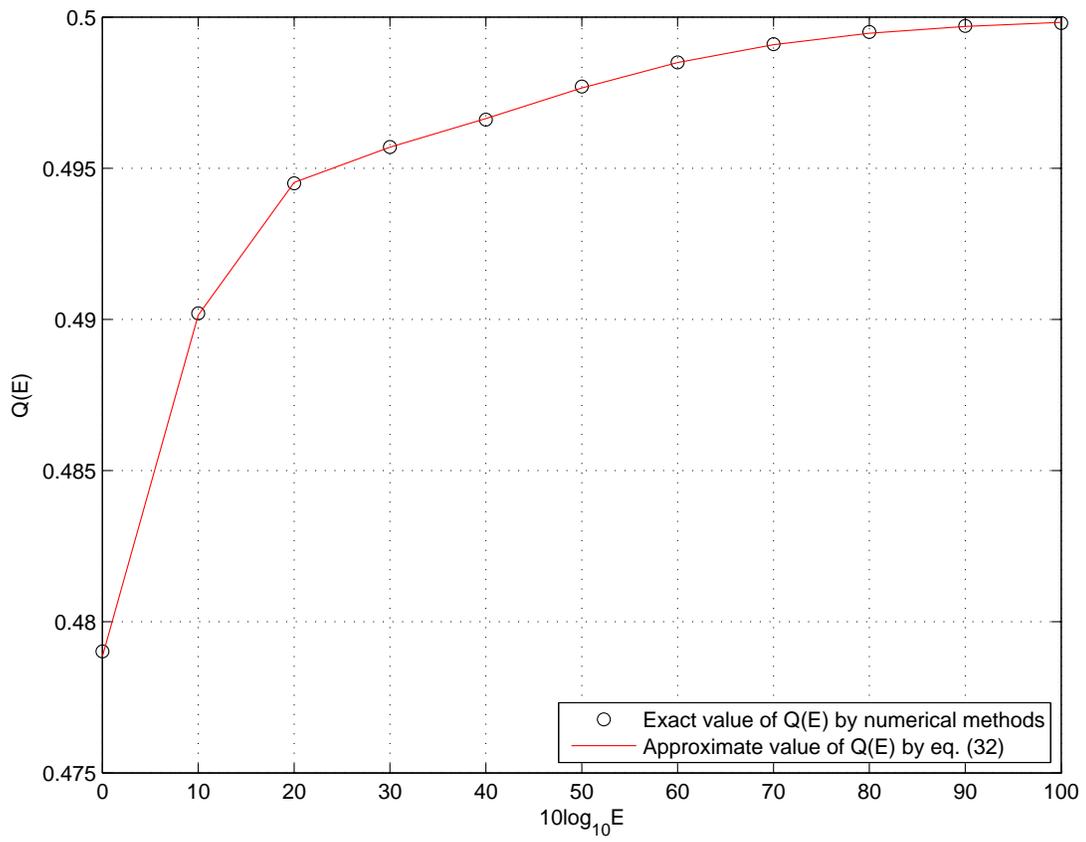}{0.9\textwidth}
\end{center}
\caption{Values of $Q(E)$, $D_1=0.5$.} \label{fig:tradeoff2}
\end{figure}

\clearpage
\newpage
\begin{figure}
\begin{center}
\subpostscript{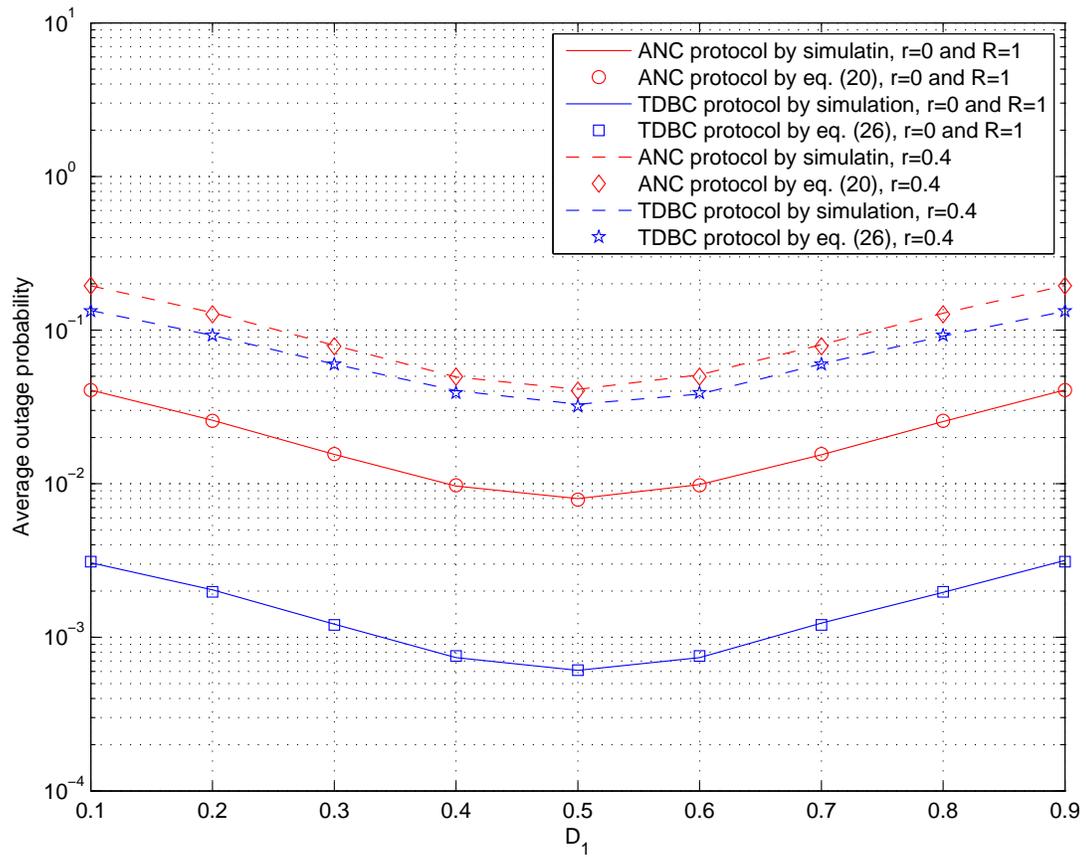}{0.9\textwidth}
\end{center}
\caption{Outage probabilities of the ANC and TDBC protocols with
equal power allocation and $\xi=1/2$, $10\log_{10}E=15$ dB.}
\label{fig:outagedis}
\end{figure}

\clearpage
\newpage
\begin{figure}
\begin{center}
\subpostscript{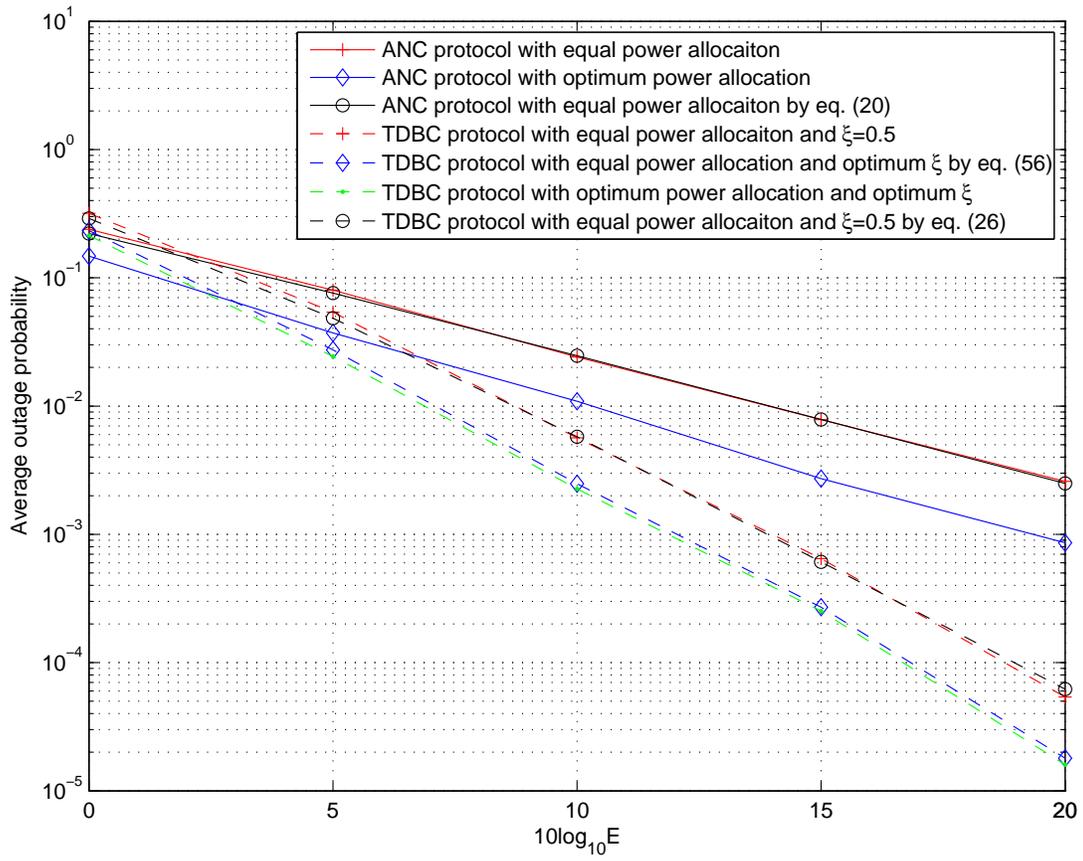}{0.9\textwidth}
\end{center}
\caption{Outage probabilities of the ANC and TDBC protocols,
$D_1=0.5$, $r=0$, and $R=1$ bps/Hz.} \label{fig:outagpower1}
\end{figure}

\clearpage
\newpage
\begin{figure}
\begin{center}
\subpostscript{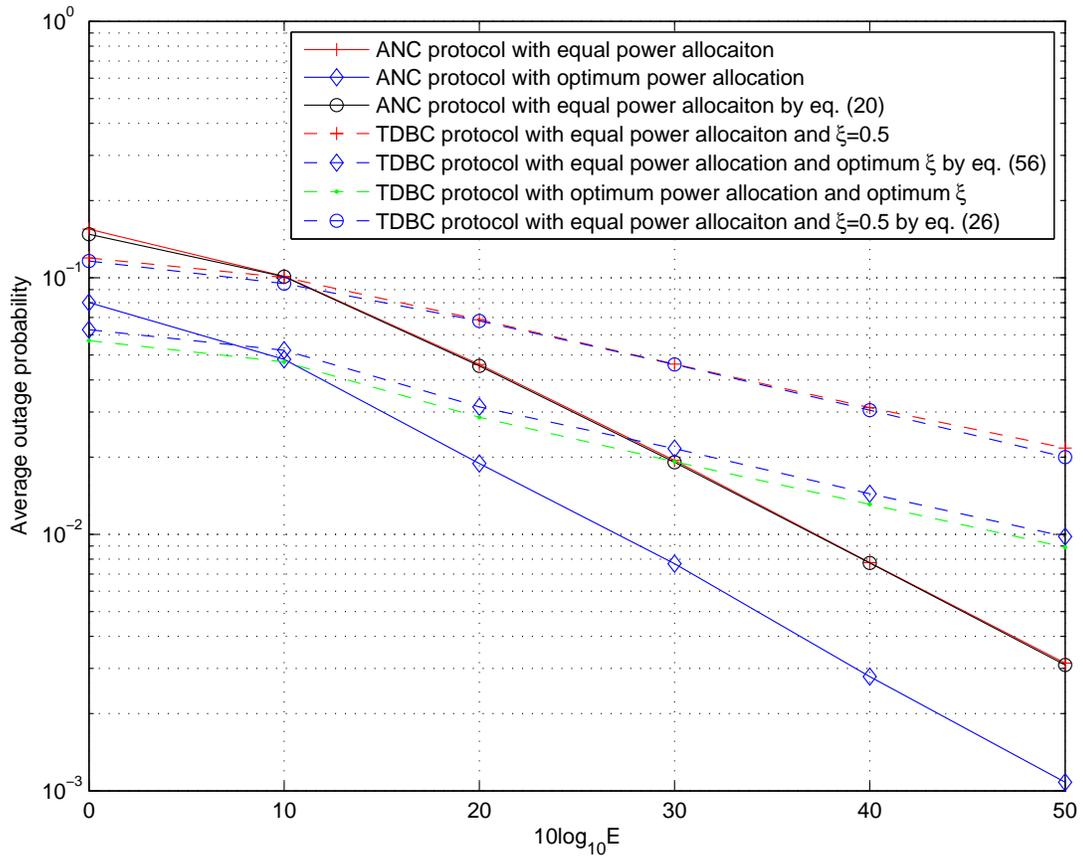}{0.9\textwidth}
\end{center}
\caption{Outage probabilities of the ANC and TDBC protocols,
$D_1=0.4$ and $r=0.6$.} \label{fig:outagpower2}
\end{figure}

\clearpage
\newpage
\begin{figure}
\begin{center}
\subpostscript{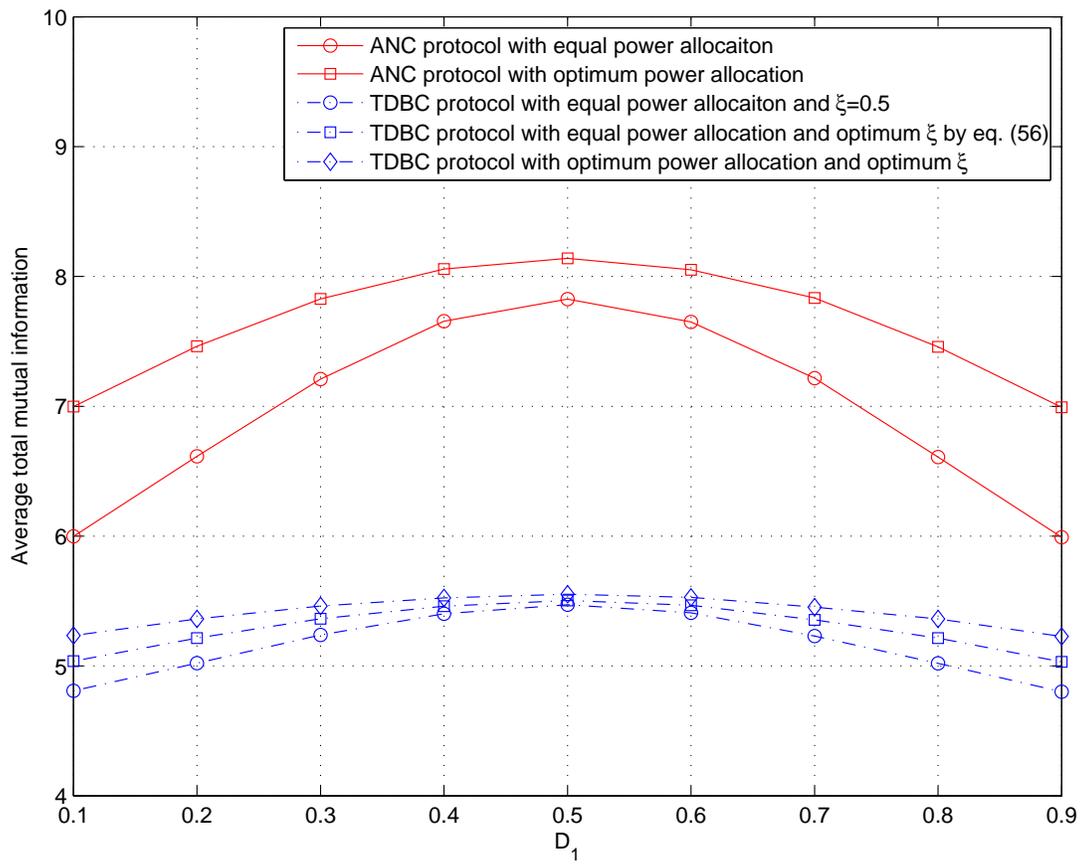}{0.9\textwidth}
\end{center}
\caption{Total mutual information of the ANC and TDBC protocols,
$10\log_{10}E=20$ dB.} \label{fig:ratedis}
\end{figure}

\end{document}